\documentclass[11pt]{article}
\usepackage{fullpage}

\usepackage{amssymb,amsmath}
\usepackage{amsthm}
\usepackage[utf8]{inputenc}
\usepackage{mathabx}
\usepackage{graphicx}
\usepackage{subfigure}
\usepackage{graphics}
\usepackage{url}
\usepackage{hyperref}
\usepackage{color}

\newtheorem{theorem}{Theorem}[section]
\newtheorem{lemma}[theorem]{Lemma}

\newcommand{\diam}{\mathsf{diam}}
\newcommand{\eps}{\varepsilon}

\newcommand{\ball}{\mathsf{ball}}

\newcommand{\capa}{\mathsf{cap}}
\newcommand{\demand}{\mathsf{dem}}

\newcommand{\SAT}{\mathrm{SAT}}

\newcommand{\density}{\mathbf{d}}

\title{Algorithms for Metric Learning via Contrastive Embeddings}

\author{
	Diego Ihara Centurion\thanks{University of Illinois at Chicago, Dept.~of Computer Science, Chicago, IL 60607. E-mail \texttt{\{dihara2,nmoham24,sidiropo\}@uic.edu}. Supported by NSF under award CAREER 1453472, and grants CCF 1815145 and CCF 1423230.}
	\and
	Neshat Mohammadi\footnotemark[1]
	\and
	Anastasios Sidiropoulos\footnotemark[1]
}

\begin{document}

\maketitle

\begin{abstract}
We study the problem of supervised learning a metric space under \emph{discriminative} constraints.
Given a universe $X$ and sets ${\cal S}, {\cal D}\subset {X \choose 2}$ of \emph{similar} and \emph{dissimilar} pairs, we seek to find a mapping $f:X\to Y$, into some target metric space $M=(Y,\rho)$, such that similar objects are mapped to points at distance at most $u$, and dissimilar objects are mapped to points at distance at least $\ell$.
More generally, the goal is to find a mapping of maximum \emph{accuracy} (that is, fraction of correctly classified pairs).
We propose approximation algorithms for various versions of this problem, for the cases of Euclidean and tree metric spaces.
For both of these target spaces, we obtain fully polynomial-time approximation schemes (FPTAS) for the case of perfect information.
In the presence of imperfect information we present approximation algorithms that run in quasi-polynomial time (QPTAS).
We also present an exact algorithm for learning line metric spaces with perfect information in polynomial time.
Our algorithms use a combination of tools from metric embeddings and graph partitioning, that could be of independent interest.
\end{abstract}

\section{Introduction}

Geometric algorithms have given rise to a plethora of tools for data analysis, such as clustering, dimensionality reduction, nearest-neighbor search, and so on; we refer the reader to \cite{DBLP:reference/cg/IndykMS17, har2011geometric, matouvsek2002lectures} for an  exposition.
A common aspect of these methods is that the underlying data is interpreted as a metric space.
That is, each object in the input is treated as a point, and the pairwise dissimilarity between pairs of objects is encoded by a distance function on pairs of points.
An important element that can critically determine the success of this data analytic framework, is the choice of the actual metric.
Broadly speaking, the area of \emph{metric learning} is concerned with methods for recovering an underlying metric space that agrees with a given set of observations (we refer the reader to \cite{shakhnarovich2005learning,kulis2013metric} for a detailed exposition).
The problem of learning the distance function is cast as an optimization problem, where the objective function quantifies the extend to which the solution satisfies the input constraints.

\textbf{Supervised vs.~unsupervised metric learning.}
The problems studied in the context of metric learning generally fall within two main categories: 
\emph{supervised} and \emph{unsupervised} learning.
In the case of unsupervised metric learning, the goal is to discover the intrinsic geometry of the data, or to fit the input into some metric space with additional structure.
A prototypical example of an unsupervised metric learning problem is dimensionality reduction, where one is given a high-dimensional point set and the goal is to find a mapping into some space of lower dimension, or with a special structure, that approximately preserves the geometry of the input (see e.g.~\cite{dasgupta2003elementary, matouvsek2010inapproximability,badoiu2007approximation}).

In contrast, the input in a supervised metric learning problem includes \emph{label constrains}, which encode some prior knowledge about the ground truth.
As an example, consider the case of a data set where experts have labeled some pairs of points as being ``similar'' and some pairs as being ``dissimilar''.
Then, a metric learning task is to find a transformation of the input such that similar points end up close together, while dissimilar points end up far away from each other.
As an illustrative example of the above definition, consider the problem of recognizing a face from a photo.
More concretely, a typical input consists of some \emph{universe} $X$, which is a set of photos of faces, together with some ${\cal S}, {\cal D}\subseteq \binom{X}{2}$.
The set ${\cal S}$ consists of pairs of photos that correspond to the same person, while ${\cal D}$ consists of pairs of photos from different people.
A typical approach for addressing this problem (see e.g.~\cite{chopra2005learning}) is to find a mapping $f:X\to \mathbb{R}^d$, for some $d\in \mathbb{N}$, such that for all similar pairs $\{x,y\}\in {\cal S}$ we have $\|f(x)-f(y)\|_2 \leq u$, and for all dissimilar pairs $\{x,y\}\in {\cal D}$ we have $\|f(x)-f(y)\|_2 \geq \ell$, for some $u,\ell>0$.
More generally, one seeks to find a mapping $f$ that maximizes the fraction of correctly classified pairs of photos (see Figure \ref{fig:example}).

\begin{figure}
\begin{center}
\scalebox{0.68}{\includegraphics{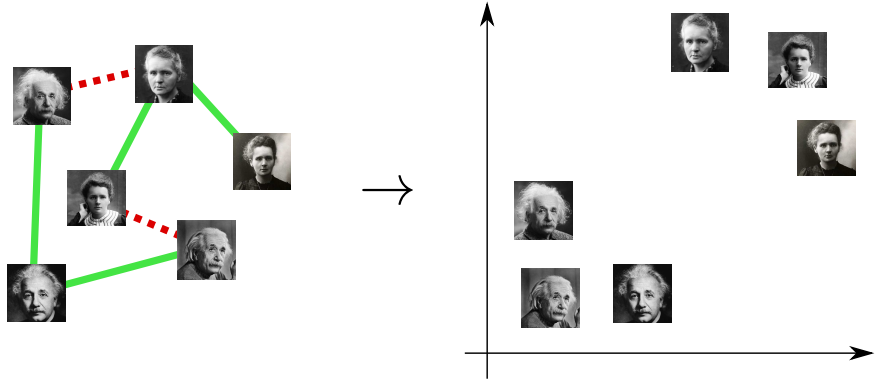}}
\caption{An illustration of the metric learning framework.
The input consists of a set of objects (here, a set of photos), with some pairs labeled as ``similar'' (depicted in green), and some pairs labeled ``dissimilar'' (depicted in red).
The output is an embedding into some metric space (here, the Euclidean plane), such that similar objects are mapped to nearby points, while dissimilar objects are mapped to points that are far from each other.
\label{fig:example}}
\end{center}
\end{figure}

\subsection{Problem formulation}

At the high level, an instance of a metric learning problem consists of some universe of objects $X$, together with some similarity information on subsets of these objects.
Here, we focus on similarity and dissimilarity constraints.
Specifically, an instance is a tuple $\phi=(X,{\cal S}, {\cal D}, u, \ell)$, where $X$ is a finite set, with $|X|=n$, ${\cal S}, {\cal D}\subset \binom{X}{2}$, which are sets of pairs of objects that are labeled as ``similar'' and ``dissimilar'' respectively, and $u,\ell>0$.
We refer to the elements of ${\cal S}\cup {\cal D}$ as \emph{constraints}.
We focus on the case where ${\cal S}\cap {\cal D}=\emptyset$, and ${\cal S}\cup {\cal D}= \binom{X}{2}$.
Let $f:X\to Y$ be a mapping into some target metric space $(Y,\rho)$.
As it is typical with geometric realization problems, we relax the definition to allow for a small multiplicative error $c\geq 1$ in the embedding.
We say that $f$ \emph{satisfies} $\{x,y\}\in {\cal S}$, if 
\begin{align}
\rho(f(x),f(y)) \leq u\cdot c, \label{eq:u}
\end{align}
and we say that it satisfies $\{x,y\}\in {\cal D}$ if
\begin{align}
\rho(f(x),f(y)) \geq \ell/c. \label{eq:l}
\end{align}
If $f$ does not satisfy some $\{x,y\}\in {\cal S}\cup {\cal D}$, then we say that it \emph{violates} it.
We refer to the parameter $c$ as the \emph{contrastive distortion} of $f$.
We also refer to $f$ as a \emph{constrastive embedding}, or \emph{$c$-embedding}.

\subsection{Our contribution}

We now briefly discuss our main contributions.

We focus on the problem of computing an embedding with low contrastive distortion, and with maximum \emph{accuracy}, which is defined to be the fraction of satisfied constraints.
Since we are assuming \emph{complete information}, that is ${\cal S}\cup {\cal D}=\binom{X}{2}$, it follows that the accuracy is equal to $k/\binom{n}{2}$, where $k$ is the number of satisfied constraints.
We remark that the setting of complete information requires a \emph{dense} set of constraints.
This case is important in applications when learning a metric space based on a  set of objects with fully-labeled pairs.
This is also the setting of other important machine learning primitives, such as Correlation Clustering (see \cite{bansal2004correlation}), which we discuss in Section \ref{sec:related}.

Our results are concerned with two main cases: \emph{perfect} and \emph{imperfect} information.
Here, the case of perfect information corresponds to the promise problem where there exists an embedding that satisfies all constraints.
On the other hand, in the case of imperfect information, no such promise is given.
As we shall see, the latter scenario appears to be significantly more challenging.
Our results are concerned with two different families of target spaces: $d$-dimensional Euclidean space, and trees.

\textbf{Learning Euclidean metric spaces.}
We begin our investigation by observing that the problem of computing a contrastive embedding into $d$-dimensional Euclidean space with perfect information is polynomial-time solvable for $d=1$, and it becomes NP-hard even for $d=2$.
The result for $d=1$ is obtained via a simple greedy algorithm, while the NP-hardness proof uses a standard reduction from the problem of recognizing unit disk graphs (see \cite{breu1998unit}).

\begin{theorem}[Learning the line with perfect information]\label{thm:perfect_info_line}
Let $\gamma=(X, {\cal S}, {\cal D}, u, \ell)$ be an instance of the metric learning problem.
Then there exists a polynomial-time algorithm which given $\gamma$, either computes an $1$-embedding $\widehat{f}:X\to \mathbb{R}$, with accuracy $1$, or correctly decides that no such embedding exists.
\end{theorem}

\begin{theorem}[Hardness of learning the plane with perfect information]\label{thm:np-hard}
Given an instance $\gamma=(X, {\cal S}, {\cal D}, u, \ell)$ of the metric learning problem, it is NP-hard to decide whether there exists an $1$-embedding $\widehat{f}:X\to \mathbb{R}^2$, with accuracy 1.
\end{theorem}

The above two results indicate that, except for what is essentially its simplest possible case, the problem is generally intractable.
This motivates the study of approximation algorithms.
Our first main result in this direction is a FPTAS for the case of perfect information, summarized in the following.

\begin{theorem}[Learning Euclidean metric spaces with perfect information]\label{thm:perfect_info}
Let $u,\ell>0$ be fixed constants.
Let $\gamma=(X, {\cal S}, {\cal D}, u, \ell)$ be an instance of the problem of learning a $d$-dimensional Euclidean metric space with perfect information, with $|X|=n$, for some $d\geq 2$.
Suppose that $\gamma$ admits an $1$-embedding $f^*:X\to \mathbb{R}^d$ with accuracy 1.
Then for any $\eps,\eps'>0$, there exists a randomized algorithm which given $\gamma$, $\eps$, and $\eps'$, computes a $(1+\eps')$-embedding $\widehat{f}:X\to \mathbb{R}^d$, with accuracy at least $1-\eps$, in time $n^{O(1)} g(\eps, \eps',d) = n^{O(1)} 2^{(\frac{d}{\eps \eps'})^{O(d^2)}}$, with high probability.
In particular, for any fixed $d$, $\eps$, and $\eps'$, the running time is polynomial.
\end{theorem}

We also obtain the following QPTAS for the case of imperfect information.

\begin{theorem}[Learning Euclidean metric spaces with imperfect information]\label{thm:imperfect_info}
Let $d\geq 1$.
Let $u,\ell>0$ be fixed constants.
Let $\gamma=(X, {\cal S}, {\cal D}, u, \ell)$ be an instance of the problem of learning a $d$-dimensional Euclidean metric space, with $|X|=n$, for some $d\geq 1$.
Suppose that $\gamma$ admits an $1$-embedding $f^*:X\to \mathbb{R}^d$ with accuracy $1-\zeta$, for some $\zeta>0$.
	Then for any $\eps,\eps'>0$, there exists an algorithm which given $\gamma$, $\eps$, $\eps'$, and $\zeta$, computes a $(1+\eps')$-embedding $\widehat{f}:X\to \mathbb{R}^d$, with accuracy at least $1-O(\zeta^{1/2} \log^{3/4} n (\log\log n)^{1/2}) - \eps$, in time $n^{O(1)} 2^{\eps^{-2} (\frac{d \log n}{\zeta \eps'})^{O(d)}}$.
In particular, for any fixed $d$, $\eps$, $\eps'$, and $\zeta$, the running time is quasi-polynomial.
\end{theorem}

\textbf{Learning tree metric spaces.}
We next consider the case of embedding into tree metric spaces.
More precisely, we are given an instance $(X, {\cal S}, {\cal D}, u, \ell)$ of the metric learning problem, and we wish to find a tree $\widehat{T}$, and an embedding $\widehat{f}:X\to \widehat{T}$, with maximum accuracy.
Note that the tree $\widehat{T}$ is not part of the input.
Our results closely resemble the ones from the case of learning Euclidean metric spaces.
Specifically, for the case of perfect information we obtain a PTAS, summarized in the following.

\begin{theorem}[Learning tree metric spaces with perfect information]\label{thm:trees_perfect_info}
Let $u,\ell>0$ be fixed constants.
Let $\gamma=(X, {\cal S}, {\cal D}, u, \ell)$ be an instance of the problem of learning a tree metric space with perfect information, with $|X|=n$.
Suppose that $\gamma$ admits an $1$-embedding $f^*:X\to V(T^*)$, for some tree $T^*$, with accuracy 1.
Then for any $\eps,\eps'>0$, there exists a randomized algorithm which given $\gamma$, $\eps$, and $\eps'$, computes some tree $\widehat{T}$, and a $(1+\eps')$-embedding $\widehat{f}:X\to V(\widehat{T})$, with accuracy at least $1-\eps$, in time $n^{O(1)} g(\eps, \eps') = n^{O(1)} 2^{1/(\eps\eps')^{O(1/\eps)}}$, with high probability.
In particular, for any fixed $\eps$ and $\eps'$, the running time is polynomial.
\end{theorem}

As in the case of learning Euclidean metric spaces, we also obtain a QPTAS for learning tree metric spaces in the presence of imperfect information, summarized in the following.

\begin{theorem}[Learning tree metric spaces with imperfect information]\label{thm:trees_imperfect_info}
Let $u,\ell>0$ be fixed constants.
Let $\gamma=(X, {\cal S}, {\cal D}, u, \ell)$ be an instance of the problem of learning a tree metric space with imperfect information, with $|X|=n$.
Suppose that $\gamma$ admits an $1$-embedding $f^*:X\to V(T^*)$, for some tree $T^*$, with accuracy $1-\zeta$, for some $\zeta>0$.
Then for any $\eps,\eps'>0$, there exists a randomized algorithm which given $\gamma$, $\eps$, $\eps'$, and $\zeta$, computes some tree $\widehat{T}$, and a $(1+\eps')$-embedding $\widehat{f}:X\to V(\widehat{T})$, with accuracy at least
$1-O(\zeta^{1/2} \log^{3/4} n (\log\log n)^{1/2}) - \eps$, in time $2^{((\log n)/(\zeta^{1/2}\eps \eps'))^{O(1/\eps)}}$.
In particular, for any fixed $\eps$, $\eps'$, and $\zeta$, the running time is quasi-polynomial.
\end{theorem}

\subsection{Overview of our techniques}

We now give a brief overview of the techniques used in obtaining our approximation algorithms for the metric learning problem.
Interestingly, our algorithms for the Euclidean case closely resemble our algorithms for tree metric spaces.

\textbf{Learning Euclidean metric spaces.}
At the high level, our algorithms for  learning Euclidean metric spaces consist of the following  steps: 

\begin{description}
\item{\textbf{Step 1: Partitioning.}}
We partition the instance into sub-instances, each admitting an embedding into a ball of small diameter.
For the case of perfect information, it is easy to show that such a partition exists, using a Lipschitz partition of $\mathbb{R}^d$ (see \cite{charikar1998approximating}).
Such a partition can be chosen so that only a small fraction of similarity constraints are ``cut'' (that is, their endpoints fall in different clusters of the partition).
However, computing such a partition is a difficult task, since we do not have an optimal embedding (which is what we seek to compute).
We overcome this obstacle by using a result of Krauthgamer and Roughgarden \cite{krauthgamer2011metric}, which allows us to compute such a partition, without access to an optimal embedding.

For the case of imperfect information, the situation is somewhat harder since there might be no embedding that satisfies all the constraints.
This implies that the partition that we use in the case of perfect information, is not guaranteed to exist anymore.
Therefore, instead of using a geometric partitioning procedure, we resort to graph-theoretic methods.
We consider the graph $G_{\cal S}$, with $V(G)=X$, and with $E(G)={\cal S}$.
Roughly speaking, we partition $G_{\cal S}$ into expanding subgraphs.
This can be done by deleting only a relatively small fraction of edges.
For each expanding subgraph, we can show that the corresponding sub-instance admits an embedding into a subspace of small diameter.

\item{\textbf{Step 2: Embedding each sub-instance.}}
Once we have partitioned the instance into sub-instances, with each one admitting an embedding into a ball of small diameter, it remains to find such an embedding.
This is done by first discretizing the target ball, and then using the theory of pseudoregular partitions (see e.g.~\cite{frieze1996regularity}) to exhaustively find a good embedding.
The discretization step introduces contrastive distortion $1+\eps'$, for some $\eps'>0$ that can be made arbitrarily small in expense of the running time. 

\item{\textbf{Step 3: Combining the embeddings.}}
Finally, we need to combine the embeddings of the sub-instances to an embedding of the original instance.
This can easily be done by ensuring that the images of different clusters are sufficiently far from each other.
Since only a small fraction of similarity constraints are cut by the original partition, it follows that the resulting accuracy is high.
\end{description}

\textbf{Learning tree metric spaces.}
Surprisingly, the above template is also used, almost verbatim, for the case of learning tree metric spaces.
The only difference is that the partitioning step now resembles the random partitioning scheme of Klein, Plotkin and Rao \cite{klein1993excluded} for minor-free graphs. 
The embedding step requires some modification, since the space of trees of bounded diameter is infinite, and thus exhaustive enumeration is not directly possible.
We resolve this issue by first showing that if there exists an embedding into a tree of small diameter, then there also exists an embedding into a \emph{single} tree, which we refer to as \emph{canonical}, and with only slightly worse accuracy.

\subsection{Related work}\label{sec:related}

\textbf{Metric embeddings.}
The theory of metric embeddings is the source of several successful methods for processing metrical data sets, which have found applications in metric learning (see, e.g.~\cite{shakhnarovich2005learning,kulis2013metric}).
Despite the common aspects of the use of metric embeddings in metric learning and in algorithm design, there are some key differences.
Perhaps the most important difference is that the metric embedding methods being used in algorithm design often correspond to unsupervised metric learning tasks.
Consequently, problems and methods that are studied in the context of metric learning, have not received much attention from the theoretical computer science community.
Many of the existing methods used in supervised metric learning tasks are based on convex optimization, and are thus limited to specific kinds of objective functions and constraints (see e.g.~\cite{weinberger2009distance,davis2007information}).

The problems considered in this paper are closely related to questions studied in the context of computing low-distortion embeddings.
For example, the problem of learning a tree metric space is closely related to the problem of embedding into a tree, which has been studied under various classical notions of distortion (see e.g.~\cite{badoiu2007approximation,DBLP:conf/focs/Bartal96, DBLP:conf/soda/NayyeriR17, alon2008ordinal,ailon2005fitting}).
Similarly, several algorithms have been proposed for the problem of computing low-distortion embeddings into Euclidean space (see e.g.~\cite{linial1995geometry,arora2008euclidean,matouvsek2010inapproximability,nayyeri2015reality, badoiu2003approximation}).
Various extensions and generalizations have also been considered for maps that approximately preserve the relative ordering of distances (see e.g.~\cite{bilu2005monotone,alon2008ordinal,buadoiu2008ordinal}).
However, we remark that all of the above results correspond to the unsupervised version of the metric learning problem, and are thus not directly applicable in the supervised setting considered here.

\textbf{Other types of supervision and embeddings.}
We note that other notions of supervision have also been considered in the metric learning literature.
For example, instead of pairs of objects that are labeled either ``similar'' or ``dissimilar'', another possibility is to have triple constraints of the form $(x,y,z)\in \binom{X}{3}$, encoding the fact that ``$x$ is more similar to $y$ than $z$''.
In this case, one seeks to find a mapping $f:X\to Y$, such that $\rho(f(x),f(y))<\rho(f(x),f(z)) - m$, for some \emph{margin} parameter $m>0$ (see e.g.~\cite{weinberger2009distance}).
However, the form of supervision that we consider is one of the most popular ones in practice.
Furthermore, it is often desirable that the mapping $f$ has a special form, such as linearity (when $X$ is a subset of some linear space), or being specified implicitly via a set of parameters, as is the case of mappings computed by neural networks.
We believe that our work can lead to further theoretical understanding of the metric learning problem under different types of supervision, and for specific classes of embeddings.

\textbf{Correlation Clustering.}
The supervised metric learning problem considered here can be thought of as a generalization of the classical Correlation Clustering problem \cite{bansal2004correlation}.
More specifically, the Correlation Clustering problem is precisely the metric learning problem with finite contrastive distortion, for the special case when the host is the uniform metric space, with $u=0$ and $\ell=1$.

\subsection{Organization}
The rest of the paper is organized as follows.
Section \ref{sec:prelim} introduces some notation and definitions.
Section \ref{sec:pseudoregular} shows how pseudoregular partitions can be used to compute near-optimal embeddings into spaces of bounded cardinality.
Section \ref{sec:euclidean_perfect} presents the algorithm for learning Euclidean spaces with perfect information.
The algorithm for learning Euclidean spaces with imperfect information is given in Section \ref{sec:euclidean_imperfect}.
The algorithms for learning tree metric spaces under perfect and imperfect information are given in Sections \ref{sec:trees_perfect} and \ref{sec:trees_imperfect} respectively.
Section \ref{sec:line_perfect} presents the exact algorithm for learning line metric spaces with perfect information.
Finally, the NP-hardness proof of learning the plane with perfect information, is given in Section \ref{sec:np-hard}

\section{Preliminaries}\label{sec:prelim}

For any non-negative integer $n$, we use the notation $[n]=\{1,\ldots,n\}$, with the convention $[0]=\emptyset$.
Let $M=(X,\rho)$ be some metric space.
We say that $M$ is a \emph{line} metric space if it can be realized as a submetric of $\mathbb{R}$, endowed with the standard distance.
For a graph $G$ and some $v\in V(G)$, we write

$N_G(v) = \{u\in V(G):\{v,u\}\in E(G)\}$.

For some $U\subset V(G)$, we denote by $E(U)$ the set of edges in $G$ that have both endpoints in $U$; that is
$E(U) = \{\{u,v\}\in E(G):\{u,v\}\subseteq U\}$.
For some $U, U'\subset V(G)$, we also write
$E(U,U') = \{\{u,v\}\in E(G): u\in U, v\in U'\}$.

\section{Pseudoregular partitions and spaces of bounded cardinality}\label{sec:pseudoregular}

In this Section we present an algorithm for computing a nearly-optimal embedding, provided that there exists a solution contained inside a ball of small cardinality.

Let $G$ be a $n$-vertex graph.
For any disjoint $A, B\subseteq V(G)$, let
$e(A,B)=|E(A,B)|$, where $E(A,B)$ denotes the set of edged with one endpoint in $A$ and one in $B$, and let
$\density(A,B)=\frac{e(A,B)}{|A|\cdot |B|}$.
Let ${\cal V}=V_1,\ldots,V_k$ be a partition of $V(G)$.
For any $i,j\in [k]$, let $\density_{i,j}=\density(V_i,V_j)$.
For any $U\subseteq V(G)$, let $U_i=U\cap V_i$.
The partition ${\cal V}$ is called \emph{$\eps$-pseudoregular} if for all disjoint $S,T\subseteq V(G)$, we have 
$\left|e(S,T) - \sum_{i\in [k]} \sum_{j\in [k]} \density_{i,j} |S\cap V_i| |T\cap V_j|\right| \leq \eps n^2$.
It is also called \emph{equitable} if for all $i,j\in [k]$, we have $||V_i|-|V_j|| \leq 1$.
We recall the following result due to Frieze and Kannan \cite{frieze1999quick} on computing pseudoregular partitions (see also \cite{frieze1996regularity,bansal2009regularity}).

\begin{theorem}[\cite{frieze1999quick}]\label{thm:pseudoregular}
There exists a randomized algorithm which given an $n$-vertex graph $G$, and $\eps,\delta>0$, computes an equitable $\eps$-pseudoregular partition of $G$ with at most $k=2^{O(\eps^{-2})}$ parts, in time $2^{{O}(\eps^{-2})} n^2 / (\eps^2 \delta^3)$, with probability at least $1-\delta$.
\end{theorem}

The following is the main result of this Section.

\begin{lemma}[Pseudoregular partitions and embeddings into spaces of small cardinality]\label{lem:pseudo_embed}
Let $\eps, \eps'>0$.
Let $\gamma_C=(C, {\cal S}, {\cal D}, u, \ell)$ be an instance of the metric learning problem, with $|C|=n$.
Let $M=(N,\rho)$ be some metric space.
Suppose that $\gamma_C$ admits a $(1+\eps')$-embedding $g^*:C\to N$ with accuracy $r^*$.
Then, there exists an algorithm which given $\gamma_C$ and $M$, computes a $(1+\eps')$-embedding $\widehat{g}:C\to N$, with accuracy at least $r^*-\eps$.
The running time is $n^{O(1)} 2^{O(|N|^5/\eps^2)}$.
\end{lemma}

\begin{proof}
Let
$c_1=\frac{\eps}{8|N|^2}$.
We define graphs $H^{\cal S}$ and $H^{\cal D}$, where $V(H^{\cal S}) = V(H^{\cal D}) = C$, 
$E(H^{\cal S}) = {\cal S} \cap \binom{C}{2}$, and
$E(H^{\cal D}) = {\cal D} \cap \binom{C}{2}$.
Using Theorem \ref{thm:pseudoregular} we compute  equitable $c_1$-pseudoregular partitions ${\cal V}^{\cal S}$ and ${\cal V}^{\cal D}$ of $H^{\cal S}$ and $H^{\cal D}$ respectively, each with at most $k=2^{O(c_1^{-2})} = 2^{O(|N|^4/\eps^2)}$ parts.
Let ${\cal V}$ be the common refinement of ${\cal V}^{\cal S}$ and ${\cal V}^{\cal D}$; that is
${\cal V} = \bigcup_{U\in {\cal V}^{\cal D}} \bigcup_{U'\in {\cal V}^{\cal S}} \{U \cap U'\}$.
We have ${\cal V}=\{V_1,\ldots,V_{k'}\}$, for some $k'\leq k^2$.

Let
$c_2=\frac{\eps}{8|N|^2 k^4} |C|$.
For each $p\in N$, and for each part $U\in {\cal V}$, let
$\tau(p,U) = |\{q\in C: g^*(q)=p\}|$.
Let also $\tau^*:N\times {\cal V} \to \{0,\ldots,|C|/c_2\}$, where 
$\tau^*(p,U) = \lfloor \tau(p,U)/c_2 \rfloor$.

Let $\widetilde{g}:C\to \mathbb{R}^d$ be obtained as follows.
We enumerate all possible $\tau':N\times {\cal V}\to \{0,\ldots,|C|/c_2\}$.
There are at most $n/c_2$ choices for each value $\tau'(p,U)$, and thus at most $(n/c_2)^{|N|}$ choices for $\tau'$ over all $p$ and $U$.
For any such choice $\tau'$, we construct a candidate $g':C\to N$ as follows:
For each $p\in N$, and for each $U\in {\cal V}$, we pick arbitrarily some set $C_{p,U}\subset U$, with $|C_{p,U}|=\tau'(p,U) \cdot c_2$, such that for any distinct $p,p'\in N$, we have $C_{p,U}\cap C_{p',U}=\emptyset$.
For all $q\in C_{p,U}$, we set $g'(q)=p$.
Let $Z$ be the set of points in $C$ that are not contained in any of the computed sets $C_{p,U}$.
We extend $g'$ to $Z$ by mapping each $q\in Z$ to some arbitrary point in $N$.
After trying all possible $\tau'$, we set $\widehat{g}$ to be the mapping $g'$ found with the highest accuracy.
This completes the description of the algorithm for computing $\widehat{g}$.
The running time of the algorithm is determined by the number of possible choices for $\tau'$, and thus for any fixed constants $u$ and $\ell$, it is at most 
$(|C|/c_2)^{|N|} |C|^{O(1)} \leq n^{O(1)} 2^{O(|N|^5/\eps^2)}$.

It remains to show that the accuracy of $\widehat{g}$ is at least $r^*-\eps$.
For any mapping $g:C\to N$, and for any $x,y\in N$, let
\[
{\cal S}_{x,y}(g) = \{\{p,q\}\in {\cal S} : g(p)=x \text{ and } g(q) = y\}.
\]
We define ${\cal S}^{\SAT}(g)$ to be the set of all constraints in ${\cal S}$ that $g$ satisfies; that is
\[
{\cal S}^{\SAT}(g) = \{\{p,q\}\in {\cal S} : \rho(g(p),g(q)) \leq u (1+\eps')\}.
\]
For any $x,y\in N$, let ${\cal S}^{\SAT}_{x,y}(g)$ be the set of all constraints in ${\cal S}$ that $g$ satisfies, with one endpoints mapped to $x$ and another to $y$, and with no endpoint in $Z$; that is
\[
{\cal S}^{\SAT}_{x,y}(g) = ({\cal S}^{\SAT}(g) \cap {\cal S}_{x,y}(g)) \setminus ((Z \times C) \cup (C\times Z)).
\]
Clearly, we have
\[
{\cal S}^{\SAT}(g) \subseteq (Z\times C)\cup (C\times Z)\cup \bigcup_{x,y\in N} {\cal S}_{x,y}^{\SAT}(g) 
\]

Let $g'$ be the function considered during the construction of $\widehat{g}$, such that $\tau'=\tau^*$.
For any $x,y\in N$, we have by pseudoregularity that 
\begin{align*}
\left| \left| {\cal S}_{x,y}^{\SAT}(g') \right| - \left| {\cal S}_{x,y}^{\SAT}(g^*) \right| \right| &\leq \left|e_{{\cal S}}(g'^{-1}(x), g'^{-1}(y)) - e_{\cal S}((g^*)^{-1}(x), (g^*)^{-1}(y))\right| \\
 &\leq \left| e_{\cal S}(g'^{-1}(x), g'^{-1}(y)) - \sum_{i,j\in [k']} \density_{i,j} |g'^{-1}(x)| |g'^{-1}(y)| \right| \\
 &~~~~ + \left| e_{\cal S}((g^*)^{-1}(x), (g^*)^{-1}(y)) - \sum_{i,j\in [k']} \density_{i,j} |(g^*)^{-1}(x)| |(g^*)^{-1}(y)| \right| \\
 &~~~~ + \left| \sum_{i,j\in [k']} \density_{i,j} |(g^*)^{-1}(x)| |(g^*)^{-1}(y)| - \sum_{i,j\in [k']} \density_{i,j} |g'^{-1}(x)| |g'^{-1}(y)|\right| \\
 &\leq 2 c_1 |C|^2 + \sum_{i,j\in [k']} \density_{i,j} (|(g^*)^{-1}(x)| |(g^*)^{-1}(y)| - (|(g^*)^{-1}(x)|-c_2) (|(g^*)^{-1}(y)|-c_2)) \\
  &\leq 2 c_1 |C|^2 + \sum_{i,j\in [k']} \density_{i,j} 2 c_2 |C| \leq 2 c_1 |C|^2 + 2(k')^2 c_2 |C| \leq 2 c_1 |C|^2 + 2 k^4 c_2 |C|
\end{align*}
We thus get
\begin{align*}
\left| \left| {\cal S}^{\SAT}(\widehat{g}) \right| - \left| {\cal S}^{\SAT}(g^*) \right| \right| &\leq
\left| \left| {\cal S}^{\SAT}(g') \right| - \left| {\cal S}^{\SAT}(g^*) \right| \right|\\
 &\leq |C| |Z| + \sum_{x,y\in N} \left| \left| {\cal S}_{x,y}^{\SAT}(g') \right| - \left| {\cal S}_{x,y}^{\SAT}(g^*) \right| \right| \\
 &\leq |C| |Z| + |N|^2 (2 c_1 |C|^2 + k^4 c_2 |C|) \\
  &\leq |C| k^2 |N| c_2 + |N|^2 (2 c_1 |C|^2 + k^4 c_2 |C|) \leq \eps |C|^2 /2
\end{align*}

By an identical argument we obtain that the total number of constraints in ${\cal D}$ that are satisfied in $g^*$ but not in $\widehat{g}$ is at most $\eps |C|^2 /2$.
Therefore the accuracy of $\widehat{g}$ is at least $r^*-\eps$, concluding the proof.
\end{proof}

\section{Learning Euclidean metric spaces with perfect information}
\label{sec:euclidean_perfect}

It this section we describe our algorithm for learning $d$-dimensional Euclidean metric spaces with perfect information.
First, we present some tools from the theory of random metric partitions, and show how they can be used to partition an instance to smaller ones, each corresponding to a subspace of bounded diameter.
The final algorithm combines this decomposition with the algorithm from Section \ref{sec:pseudoregular} for learning metric spaces of bounded diameter.

\subsection{Euclidean spaces and Lipschitz partitions}

We introduce the necessary tools for randomly partitioning a metric space, and use them to obtain a decomposition of the input into pieces, each admitting an embedding into a ball in $\mathbb{R}^d$ of small diameter.

Let $M=(X,\rho)$ be a metric space.
A partition $P$ of $X$ is called \emph{$\Delta$-bounded}, for some $\Delta>0$, if every cluster in $P$ has diameter at most $\Delta$.
Let ${\cal P}$ be a probability distribution over partitions of $X$.
We say that ${\cal P}$ is \emph{$(\beta, \Delta)$-Lipschitz}, for some $\beta, \Delta>0$, if the following conditions are satisfied:
\begin{itemize}
\item
Any partition in the support of ${\cal P}$, is $\Delta$-bounded.
\item
For all $p,q\in X$, $\Pr_{P\sim {\cal P}}[P(x)\neq P(y)] \leq \beta \cdot \rho(x,y)$.
\end{itemize}

\begin{lemma}[\cite{charikar1998approximating}]\label{lem:Rd-decompose}
For any $d\geq 1$, and any $\Delta>0$, we have that $d$-dimensional Euclidean space admits a $(O(\sqrt{d}/\Delta), \Delta)$-Lipschitz distribution.
\end{lemma}

Let $G_{\cal S}$ be the graph with vertex set $V(G_{\cal S})=X$, and edge set $E(G_{\cal S}) = {\cal S}$.
We use $\rho_{\cal S}$ to denote the shortest-path distance of $G_{\cal S}$, where the length of each edge is set to $u$.

\begin{lemma}\label{lem:exact_diam}
Let $X'\subset X$, such that the subgraph of $G_{\cal S}$ induced on $X'$ (i.e.~$G_{\cal S}[X']$) is connected.
Suppose further that $\diam(f^*(X'))\leq \Delta$, for some $\Delta>0$.
Then, for any $p,q\in X'$, we have
$\rho_{\cal S}(p,q) \leq 2 u (1+4\Delta/u)^{d}$.
\end{lemma}

\begin{proof}
Let $G'=G_{\cal S}[X']$.
Let $Q=x_1,\ldots,x_L$ be the shortest path between $p$ and $q$ in $G'$, with $x_1=p$, $x_L=q$.
We first claim that for all $i,j\in \{1,\ldots,L\}$, with $i<j-1$, we have that 
\begin{align}
\|f^*(x_i), f^*(x_j)\|_2 > u. \label{eq:exact_ij}
\end{align}
Indeed, note that if $\|f^*(x_i), f^*(x_j)\|_2 \leq u$, then since $\gamma$ has full information, i.e.~${\cal S}\cup {\cal D}=\binom{X}{2}$, it must be that $\{x_i,x_j\}\in {\cal S}$, and thus the edge $\{x_i,x_j\}$ is present in the graph $G'$.
This implies that $\rho_{G'}(x_i,x_j)=u$, which contradicts the fact that $Q$ is a shortest path.
We have thus established \eqref{eq:exact_ij}.

For each $i\in \{1,\ldots,L\}$, let $B_i$ be the ball in $\mathbb{R}^d$ centered at $f^*(x_i)$ and of radius $u/2$.
By \eqref{eq:exact_ij} we get that the balls $B_2,B_4,\ldots,B_{2\lfloor L/2 \rfloor}$ are pairwise disjoint.
Since $\diam(f^*(X')) \leq \Delta$, it follows that there exists some $x^*\in \mathbb{R}^d$, such that 
$f^*(X') \subset \ball(x^*, 2\Delta)$.
Therefore
\begin{align}
\bigcup_{i=1}^{\lfloor L/2\rfloor} B_{2i} \subseteq \bigcup_{i=1}^{\lfloor L/2\rfloor} \ball(x_{2i}, u/2) \subseteq \ball(x^*, 2\Delta+u/2). \label{eq:exact_balls}
\end{align}
Recall that the volume of the ball of radius $r$ in $\mathbb{R}^d$ is equal to $V_d(r)=\frac{\pi^{d/2}}{\Gamma(1+d/2)}r^d$.
From \eqref{eq:exact_balls} we get
$\lfloor L/2 \rfloor \cdot V_d(u/2) \leq V_d(2\Delta+u/2)$,
and thus
$\rho_{\cal S}(p,q) \leq \rho_{G'}(p,q) = u \cdot (L-1) \leq 2 u  \frac{V_d(2\Delta+u/2)}{V_d(u/2)}= 2 u (1+4\Delta/u)^d$,
which concludes the proof.
\end{proof}

We now show that the shortest-path metric of the similarity constraint graph admits a Lipschitz distribution.

\begin{lemma}\label{lem:exact_Lip}
For any $\Delta>0$,
the metric space $(X,\rho_{\cal S})$ admits a $(O(\sqrt{d}/\Delta), u(1+4\Delta/u)^d)$-Lipschitz distribution.
\end{lemma}

\begin{proof}
Let $f^*:X\to \mathbb{R}^d$ be some mapping with accuracy $1$.
By Lemma \ref{lem:Rd-decompose}
there exist some $(O(\sqrt{d}/\Delta), \Delta)$-Lipschitz distribution, ${\cal P}$, of $(\mathbb{R}^d, \|\cdot\|_2)$.
Let $P$ be a random partition sampled from ${\cal P}$.
We define a partition $P'$ of $X$ as follows.
Let $G_{\cal S}'$ be the graph obtained from $G_{\cal S}$ by deleting all edges whose endpoints under $f^*$ are in different clusters of $P$.
That is, $V(G_{\cal S}')=X$, and 
$E(G_{\cal S}') = \{\{p,q\} \in {\cal S} : P(f^*(p))=P(f^*(q))\}$,
where $P(p)$ denotes the cluster of $P$ containing $p$.
We define $P'$ to be the set of connected components of $G_{\cal S}'$.
Let ${\cal P}'$ be the induced distribution over partitions of $(X,\rho_{\cal S})$.
It remains to show that ${\cal P}'$ is $(O(\sqrt{d}/\Delta), u(1+4\Delta/u)^d)$-Lipschitz.

Let us first bound the probability of separation.
Let $p,q\in X$.
Let $x_1,\ldots,x_t$, with $x_1=p$, $x_t=q$, be a shortest-path in $G_{\cal S}$ between $p$ and $q$.
Since ${\cal P}$ is $(O(\sqrt{d}/\Delta), \Delta)$-Lipschitz,
 it follows that 
$\Pr[f^*(p)\neq f^*(q)] \leq \frac{O(\sqrt{d})}{\Delta} \|f^*(p)-f^*(q)\|_2 \leq \frac{O(\sqrt{d})}{\Delta} \sum_{i=1}^{t-1} \|f^*(x_i)-f^*(x_{i+1})\|_2 \leq \frac{O(\sqrt{d})}{\Delta} \sum_{i=1}^{t-1} \rho_S(x_i,x_{i+1}) = \frac{O(\sqrt{d})}{\Delta} \rho_{\cal S}(p,q)$.

Finally, let us bound the diameter of the clusters in $P'$.
Let $C$ be a cluster in $P'$.
By construction, $\diam(f^*(C))\leq \Delta$, and $G_{\cal S}[C]$ is a connected subgraph, thus by Lemma \ref{lem:exact_diam} we obtain that the diameter of $C$ in the metric space $(X,\rho)$ is at most $u(1+4\Delta/u)^d$.
Thus ${\cal P}'$ is $(O(\sqrt{d}/\Delta), u(1+4\Delta/u)^d)$-Lipschitz, concluding the proof.
\end{proof}

The following result allows us to sample from a Lipschitz distribution, without additional information about the geometry of the underlying space.

\begin{lemma}[\cite{krauthgamer2011metric}]\label{lem:Lip_approx}
Let $M=(X, \rho)$ be a metric space, and let $\Delta>0$.
Suppose that $M$ admits a $(\beta, \Delta)$-Lipschitz distribution, for some $\beta>0$.
Then there exists a randomized polynomial-time algorithm, which given $M$ and $\Delta$, outputs a random partition $P$ of $X$, such that $P$ is sampled from a $(2\beta, \Delta)$-Lipschitz distribution ${\cal P}$.
\end{lemma}

\subsection{The algorithm: Combining Lipschitz and pseudoregular partitions}

We now present our final algorithm for learning Euclidean metric spaces with perfect information.
First, we show that, using the algorithm from Section \ref{sec:pseudoregular}, we can learn $d$-dimensional Euclidean metric spaces of bounded diameter.
The final algorithm combines this result with an efficient procedure for decomposing the instance using a random Lipschitz partition.

\begin{lemma}[Computing Euclidean embeddings of small diameter]\label{lem:pseudoregular_embedding_euclidean}
Let $\eps, \eps'>0$.
Let $\gamma_C=(C, {\cal S}, {\cal D}, u, \ell)$, with $|C|=n$, be an instance of the metric learning problem.
Suppose that $\gamma_C$ admits an embedding $g^*:C\to \mathbb{R}^d$ with accuracy $r^*$, such that $\diam(g^*(C)) \leq \Delta$, for some $\Delta>0$.
Then, there exists an algorithm which given $\gamma_C$ and $\Delta$, computes a $(1+\eps')$-embedding $\widehat{g}:C\to\mathbb{R}^d$, with accuracy at least $r^*-\eps$.
Furthermore for any fixed $u$ and $\ell$, the running time is $T(n, d, \Delta, \eps, \eps') = 2^{\eps^{-2} (\Delta \sqrt{d}/\eps')^{O(d)}}$.
\end{lemma}

\begin{proof}
Let
$c_1=\frac{\eps'}{2\sqrt{d}}\min\{\ell,u\}$.
Since $\diam(g^*(C))\leq \Delta$, it follows that there exists some ball $B\subset \mathbb{R}^d$ of radius $2\Delta$ such that $g^*(C)\subset B$.
Let $c_1\mathbb{Z}^d$ denote the integer lattice scaled by a factor of $c_1$.
Let $N=B\cap (c_1\mathbb{Z}^d)$.
Note that $|N| \leq (2\Delta/c_1)^{d}$.
Let $g':C\to \mathbb{R}^d$ be defined such that for each $p\in C$, we have that $g'(p)$ is a nearest neighbor of $g^*(p)$ in $N$, breaking ties arbitrarily.
For any $\{p,q\}\in {\cal S}\cap \binom{C}{2}$, we have
$\|g'(p)-g'(q)\|_2 \leq \|g'(p)-g^*(p)\|_2 + \|g^*(p)-g^*(q)\|_2 + \|g^*(q)-g'(q)\|_2\\
 \leq u + \sqrt{d} c_1 \leq u(1+\eps')$.
Similarly, for any $\{p,q\}\in {\cal D}\cap \binom{C}{2}$, we have
$\|g'(p)-g'(q)\|_2 \geq -\|g'(p)-g^*(p)\|_2 + \|g^*(p)-g^*(q)\|_2 - \|g^*(q)-g'(q)\|_2\\
 \geq \ell - \sqrt{d} c_1 \leq \ell/(1+\eps')$,
and thus $g'$ is a $(1+\eps')$-embedding, with accuracy $r^*$.

Let $M=(N,\|\cdot\|_2)$; that is, the metric space on $N$ endowed with the Euclidean distance.
By Lemma \ref{lem:pseudo_embed} we can therefore compute a $(1+\eps')$-embedding $\widehat{g}:C \to N$, with accuracy $r^*-\eps$, in time 
$n^{O(1)} 2^{O(|N|^5/\eps^2)} = 2^{\eps^{-2} (\Delta \sqrt{d}/\eps')^{O(d)}}$.
\end{proof}

We are now ready to prove the main result in this Section.

\begin{proof}[Proof of Theorem \ref{thm:perfect_info}]
We first construct the graph $G_{\cal S}$ as above.
Let $\Delta=c \sqrt{d} u / \eps$, for some universal constant $c>0$.
By Lemma \ref{lem:exact_Lip} the metric space $(X, \rho_{\cal S})$ admits some $(O(\sqrt{d}/\Delta), \Delta')$-Lipschitz distribution, for some $\Delta'=u(1+4\Delta/u)^d$.
By Lemma \ref{lem:Lip_approx} we can compute, in polynomial time, some random partition $P$ of $X$, such that $P$ is distributed according to some $(O(\sqrt{d}/\Delta), \Delta')$-Lipschitz distribution ${\cal P}$.

Let ${\cal S}'=\{\{p,q\}\in {\cal S} : P(p)\neq P(q)\}$.
Since ${\cal P}$ is $(O(\sqrt{d}/\Delta), \Delta')$-Lipschitz, it follows by the linearity of expectation that 
\begin{align}
\mathbb{E}[|{\cal S}'|] &= \sum_{\{p,q\}\in {\cal S}} \Pr[P(p)\neq P(q)] \leq \sum_{\{p,q\}\in {\cal S}} O(\sqrt{d}) \frac{\rho_{\cal S}(p,q)}{\Delta} = O(\sqrt{d}) \frac{u}{\Delta} |{\cal S}| \leq \eps |{\cal S}|/4, \label{eq:S_prime}
\end{align}
where the last inequality holds for some large enough constant $c>0$.

Fix some optimal embedding $f^*:X\to \mathbb{R}^d$, that satisfies all the constraints in $\gamma$.
Let $C$ be a cluster in $P$.
Since $P$ is $\Delta'$-bounded, it follows that the diameter of $C$ in $(X,\rho_{\cal S})$ is at most $\Delta'$.
Thus, for any $p,q\in C$, there exists in $G_{\cal S}$ some path $Q=x_0,\ldots,x_L$ between $p$ and $q$, with $L\leq \Delta'/u$ edges. 
Thus
$\|f^*(p)-f^*(q)\|_2 \leq \sum_{i=0}^{L-1} \|f^*(x_i)-f^*(x_{i+1})\|_2 \leq L\cdot u \leq \Delta'$,
which implies that $\diam(f^*(C)) \leq \Delta'$.
Furthermore $f^*$ has accuracy 1.

Let ${\cal S}_C={\cal S}\cap \binom{C}{2}$, ${\cal D}_C={\cal D}\cap \binom{C}{2}$.
Then $\gamma_C=(C, {\cal S}_C, {\cal D}_C, u, u)$ is an instance of non-linear metric learning in $d$-dimensional Euclidean space that admits a solution $f^*:C\to \mathbb{R}^d$ with accuracy 1, with $\diam(f^*(C)) \leq \Delta'$.
Therefore, for each $C\in P$, using Lemma \ref{lem:pseudoregular_embedding_euclidean} we can compute some $(1+\eps')$-embedding $f_C:C\to \mathbb{R}^d$,
with accuracy at least $1-\eps/2$,
in time
$T(n, d, \Delta', \eps/2, \eps') = n^{O(1)} 2^{(\frac{d}{\eps \eps'})^{O(d^2)}}$.
We can now combine all these maps $f_C$ into a single map $\widehat{f}:X\to \mathbb{R}^d$, by translating each image $f_C(C)$ such that for any distinct $C,C'\in P$, for any $p\in C$, $p'\in C'$, we have $\|f_C(p)-f_{C'}(p')\|_2 \geq u$.
For any $p\in C$, we set $f(p)=f_C(p)$, where $C$ is the unique cluster in $P$ containing $p$.
It is immediate that any constraint $\{p,q\}\in {\cal S}$, with $p$ and $q$ in different clusters is violated by $f$.
The total number of such violations is $|{\cal S}'|$.
By Markov's inequality and \eqref{eq:S_prime}, these violations are at most $\eps|{\cal S}|/2$, with probability at least $1/2$.
Similarly, any $\{p,q\}\in {\cal D}$, with $p$ and $q$ in different clusters is satisfied by $f$.
All remaining violations are on constraints with both endpoints in the same cluster of $P$.
Since the accuracy of each $f_C$ is at least $1-\eps$, it follows that the total number of these violations is at most $\eps (|{\cal S}|+|{\cal D}|)/2$.
Therefore, the total number of violations among all constraints is at most $\eps (|{\cal S}|+|{\cal D}|)/2 + \eps |{\cal S}|/2 \leq \eps (|{\cal S}| + |{\cal D}|)$, with probability at least $1/2$.
In other words, the accuracy of $f$ is at least $1-\eps$, with probability at least $1/2$.
The success probability can be increased to $1-1/n^c$, for any constant $c>0$, by repeating the algorithm $O(\log n)$ times and returning the best embedding found.
Finally, the running time is dominated by the computation of the maps $f_C$, for all clusters $C$ in $P$.
Since there are at most $n$ such clusters, the total running time is $n^{O(1)} 2^{(\frac{d}{\eps \eps'})^{O(d^2)}}$, concluding the proof.
\end{proof}

\section{Learning Euclidean metric spaces with imperfect information}
\label{sec:euclidean_imperfect}

In this Section we describe our algorithm for learning  $d$-dimensional Euclidean metric spaces with imperfect information.
We first obtain some preliminary results on graph partitioning via sparse cuts.
We next show that for any instance whose similarity constraints induce an expander graph, the optimal solution must be contained inside a ball of small diameter.
Our final algorithm combines these results with the algorithm from Section \ref{sec:pseudoregular} for embedding into host metric spaces of bounded cardinality.

Fix some instance $\gamma=(X, {\cal S}, {\cal D}, u, \ell)$.
For the remainder of this Section we define the graphs $G_{\cal S}=(X, E_{\cal S})$, $G_{\cal D}=(X, E_{\cal D})$, and $G_{{\cal S}\cup {\cal D}}=(X, E_{{\cal S}\cup {\cal D}})$,
where $E_{\cal S}={\cal S}$, $E_{\cal D}={\cal D}$, and $E_{{\cal S}\cup {\cal D}}={\cal S}\cup {\cal D}$.

\subsection{Well-linked decompositions}

We recall some results on partitioning graphs into well-connected components.
An instance to the Sparsest-Cut problem consists of a graph $G$, with each edge $\{x,y\}\in E(G)$ having capacity $\capa(x,y)\geq 0$, and each pair $x,y\in V(G)$ having \emph{demand} $\demand(x,y)\geq 0$.
The goal is to find a cut $(U,\widebar{U})$ of minimum \emph{sparsity}, denoted by $\phi(U)$, which is defined as
$\phi(U) = \frac{\capa(U, \widebar{U})}{\demand(U, \widebar{U})}$,
where 
$\capa(U, \widebar{U}) = \sum_{\{x,y\}\in E(U, \widebar{U})} \capa(x,y)$,
$\demand(U, \widebar{U}) = \sum_{\{x,y\}\in E(U, \widebar{U})} \demand(x,y)$,
and $E(U,\widebar{U})$ denotes the set of edges between $U$ and $\widebar{U}$.
We recall the following result of Arora, Lee and Naor \cite{arora2008euclidean} on approximating the Sparsest-Cut problem (see also \cite{arora2009expander}).

\begin{theorem}[\cite{arora2008euclidean}]\label{thm:ALN}
There exists a polynomial-time $O(\sqrt{\log n} \log\log n)$-approximation for the Sparsest-Cut problem on $n$-vertex graphs.
\end{theorem}

\begin{lemma}\label{lem:well-linked}
Let $\alpha>0$.
There exists a polynomial-time algorithm which computes some $E'\subset {\cal S}\cup {\cal D}$, 
with $|E'|\leq \alpha |{\cal S}\cup {\cal D}|$, 
such that for every connected component $C$ of $G_{{\cal S}} \setminus E'$, 
and any $U\subset C$, 
we have
$|E_{{\cal S}}(U, C\setminus U)| = \Omega\left(\frac{\alpha }{\log^{3/2} n \log \log n}\right) |U| \cdot |C\setminus U|$.
\end{lemma}

\begin{proof}
We compute $E'$ inductively starting with $E'=\emptyset$.
We construct an instance of the Sparsest-Cut problem on graph $G_{{\cal S}}$, where for any $\{x,y\}\in {\cal S}$ we have $\capa(x,y)=1$, and for any $x,y\in X$, we have 
$\demand(x,y) = 1$.
Let $\chi=\frac{\alpha}{c \log^{3/2} n \log \log n}$, for some universal constant $c>0$ to be specified.
If there exists a cut of sparsity at most $\chi$, then by Theorem \ref{thm:ALN} we can compute a cut $(U, \widebar{U})$ of sparsity  $O(\chi \sqrt{\log n}\log\log n)$.
We add to $E'$ all the edges in $E_{{\cal S}}(U, \widebar{U})$, and we recurse on the subgraphs $G_{{\cal S}}[U]$ and $G_{{\cal S}}[\widebar{U}]$.
If no cut with the desired sparsity exists, then we terminate the recursion.
For each cut $(U,\widebar{U})$ found, the edges in $E_{\cal S}(U,\widebar{U})$  that were added to $E'$ are charged to the edges in $E_{{\cal S} \cup {\cal D}}(U, \widebar{U})$.
In total, we charge only $O(\chi \sqrt{\log n}\log\log n)$-fraction of all edges, and thus
$|E'| = O(\chi n \log^{3/2} n \log\log n)\leq \alpha |{\cal S}\cup {\cal D}|$, where the last inequality follows for some sufficiently large constant $c>0$, concluding the proof.
\end{proof}

\subsection{Isoperimetry and the diameter of embeddings}

Here we show that if the graph $G_{\cal S}$ of similarity constraints is ``expanding'' relative to $G_{{\cal S}\cup {\cal D}}$, then there exists an embedding with near-optimal accuracy, that has an image of small diameter.
Technically, we require that in any cut in $G_{{\cal S}\cup {\cal D}}$, a significant fraction of its edges are in ${\cal S}$.
Intuitively, this condition forces almost all points to be mapped within a small ball.
The next Lemma formalizes this statement.

\begin{lemma}\label{lem:isoperimetry_graph}
Let $F\subset {\cal S}$, with $|F| \leq \zeta \binom{|X|}{2}$, for some $\zeta>0$.
Suppose that for all $U\subset X$, we have
\begin{align}
|E_{{\cal S}}(U, X\setminus U)| \geq \alpha'\cdot |U| \cdot |X\setminus U|, \label{eq:iso_mera}
\end{align}
for some $\alpha'>0$.
Then there exists $J\subseteq X$, satisfying the following conditions:
\begin{description}
\item{(1)}
No edge in $F$ has both endpoints in $J$; that is, $F\cap E_{\cal S}(J) = \emptyset$.

\item{(2)}
$|E_{{\cal S}\cup {\cal D}} (J) | \geq (1-\zeta/\alpha') \binom{|X|}{2}$.

\item{(3)}
For all $U\subset J$, we have
\begin{align}
|E_{\cal S}(U, J\setminus U)| = \Omega(\alpha') |U| \cdot |J\setminus U|.  \label{eq:CU} 
\end{align}
\end{description}
\end{lemma}

\begin{proof}
Let $C$ be the largest connected component of $G_{\cal S}\setminus F$, and let $C'=X\setminus C$.
By \eqref{eq:iso_mera} we get
\begin{align}
|E_{{\cal S}\cup {\cal D}}(C, C')| &= |C| \cdot |C'| = O\left(1/\alpha'\right) |E_{{\cal S}}(C, C')| = O\left(1/\alpha'\right) |F| = O\left( \zeta / \alpha'\right) |X|^2. \label{eq:nyxta1}
\end{align}
Since $G_{{\cal S}\cup {\cal D}}$ is a complete graph, it follows from \eqref{eq:nyxta1} that the number of edges not in $G_{{\cal S}\cup {\cal D}}[C]$ is at most $O(\zeta/\alpha') |X|^2$, and therefore
\begin{align}
|E_{{\cal S}\cup {\cal D}}(C)| &\geq (1-O(\zeta/\alpha')) \binom{|X|}{2}. \label{eq:C_big_mera}
\end{align}

Next, we compute some $J\subseteq C$ such that for all $U\subset C$, 
\begin{align}
E_{\cal S}(U, J\setminus U) \geq c' \alpha' \cdot |U|\cdot |J\setminus U|, \label{eq:mera1}
\end{align}
for some universal constant $c'$ to be determined.
This is done as follows.
If there is no $U\subset C$ violating \eqref{eq:mera1}, then we set $J=C$.
Otherwise, pick some $U\subset C$, such that
\begin{align}
E_{\cal S}(U, C\setminus U) < c' \alpha' \cdot |U|\cdot |C\setminus U|. \label{eq:mera2}
\end{align}

Assume w.l.o.g.~that $|U|\leq |C\setminus U|$.
We delete $U$ and we recurse on $C\setminus U$.
All deleted edges are charged to the edges in $F$ that are incident to $U$.
It follows by \eqref{eq:mera2} and \eqref{eq:iso_mera} that, for some sufficiently small constant $c'>0$, at least a $2/3$-fraction of the edges incident to $U$ are in $F$.
Thus $|F\cap E_{\cal S}(U,X\setminus U)| = \Omega(\alpha') |U| \cdot |X\setminus U|$.
We charge the deleted edges (that is, the edges inside the deleted component $U$ and in $E_{\cal S}(U, X\setminus U)$) to the edges in $F\cap E_{\cal S}(U, X\setminus U)$, with each edge thus receiving $O(1/\alpha')$ units of charge.
When we recurse on $C\setminus U$, the part of $F$ that was incident to $U$ is replaced by $E_{\cal S}(U,C\setminus U)$, and thus it decreases in size by at least a factor of $2$.
Therefore, the total charge that we pay throughout the construction is at most $O(1/\alpha') |F| = O(\zeta/\alpha') |X|^2$, which is an upper bound on the number of all deleted edges.
This completes the construction of $J$.
Combining with \eqref{eq:C_big_mera} we get
$|E_{{\cal S}\cup {\cal D}}(J)| \geq |E_{{\cal S}\cup {\cal D}}(C)| - O(\zeta/\alpha')|X|^2 \geq (1-O(\zeta/\alpha')) \binom{|X|}{2}$,
which concludes the proof.
\end{proof}

\begin{lemma}\label{lem:isoperimetry}
Let $M=(Y,\rho)$ be any metric space.
Suppose that $\gamma$ admits an embedding $f^*:X\to Y$ with accuracy $1-\zeta$, for some $\zeta>0$.
Suppose further that for all $U\subset X$, we have
\begin{align}
E_{{\cal S}}(U, X\setminus U) \geq \alpha' \cdot |U| \cdot |X\setminus U|, \label{eq:iso_embed}
\end{align}
for some $\alpha'>0$.
Then there exists an embedding $f':X\to Y$, with accuracy at least $1-O(\zeta/\alpha')$, such that $\diam_M(f'(X)) = O\left( \frac{u \log n}{\alpha'}\right)$.
\end{lemma}

\begin{proof}
Let $J\subseteq X$ be given by Lemma \ref{lem:isoperimetry_graph}, where we set $F$ to be the set of constraints in ${\cal S}$ that are violated by $f^*$.
We define $f':X\to Y$ by mapping each $x\in J$ to $f^*(x)$, and each $x\in X\setminus J$ to some arbitrary point in $f^*(J)$.
All constraints in $E_{\cal S}(J)$ are satisfied by $f^*$, and thus also by $f'$.
It follows that $f'$ can only violate constraints that are violated by $f^*$, and constraints that are not in $E_{{\cal S}\cup {\cal D}}(J)$.
Therefore the accuracy of $f'$ is at least 
$1 - \zeta - O(\zeta/\alpha') = 1-O(\zeta/\alpha')$.

By Lemma \ref{lem:isoperimetry_graph} we have that 
for all $U\subset J$, 
$|E_{\cal S}(U, J\setminus U)| = \Omega(\alpha') |U| \cdot |J\setminus U|$.
Therefore, the combinatorial diameter (that is, the maximum number of edges in any shortest path) of $G_{\cal S}[J]$ is at most $O((\log n)/\alpha')$.
For all $\{x,y\}\in E_{\cal S}(J)$ we have $\rho(f'(x),f'(y))= \rho(f^*(x),f^*(y)) \leq u$, it follows that 
$\diam_M(f'(X)) = \diam_M(f'(J)) = O(\frac{u \log n}{\alpha'})$.
\end{proof}

\subsection{The algorithm}
We are now ready to prove the main result of this Section.

\begin{proof}[Proof of Theorem \ref{thm:imperfect_info}]
Fix some optimal embedding $f^*:X\to \mathbb{R}^d$, with accuracy $1-\zeta$.
Let $E'\subset {\cal S} \cup {\cal D}$ be the set of constraints computed by the algorithm in Lemma \ref{lem:well-linked}.
We have 
$|E'| \leq \alpha |{\cal S}\cup {\cal D}|$,
for some $\alpha>0$ to be determined.

For each connected component $C$ of $G_{{\cal S}} \setminus E'$, let $\gamma_C=(C,{\cal S}_C,{\cal D}_C,u,\ell)$ be the restriction of $\gamma$ on $C$;
that is, ${\cal S}_C={\cal S}\cap \binom{C}{2}$, and ${\cal S}_D={\cal D}\cap \binom{C}{2}$.
Let $f^*_C$ be the restriction of $f^*$ on $C$, and let the accuracy of $f^*_C$ be $1-\zeta_C$, for some $\zeta_C\in [0,1]$.
Let $F_C$ be the set of constraints in $E_{\cal S}(C)$ that are violated by $f^*_C$.
By Lemma \ref{lem:isoperimetry} it follows that there exists an embedding $f'_C:X\to \mathbb{R}^d$, with accuracy $1-O(\zeta_C/\alpha')$, and such that $\diam(f'_C(C)) = O(\frac{u\log n}{\alpha'})$, for some $\alpha'=\Omega(\frac{\alpha}{\log^{3/2} n \log\log n})$.

Using Lemma \ref{lem:pseudoregular_embedding_euclidean} we can compute a $(1+\eps')$-embedding $\widehat{f}_C:C\to \mathbb{R}^d$ with accuracy at least $1-O(\zeta_C/\alpha)-\eps$, in time $n^{O(1)} 2^{\eps^{-2} (\frac{\Delta \sqrt{d}}{\eps'})^{O(d)}}$.
We can combine all the embeddings $\widehat{f}_C$ into a single embedding $\widehat{f}:X\to \mathbb{R}^d$ by translating the images of $\widehat{f}_C(C)$ and $\widehat{f}_{C'}(C')$ in $\mathbb{R}^d$ to that their  distance is at least $\ell$, for all distinct components $C$, $C'$.
It is immediate that the resulting embedding $\widehat{f}$ can only violate the constraints in $E'$, and the constraints violated by all the $\widehat{f}_C$.
Thus, the accuracy of $\widehat{f}$ is at least $1-\alpha - O(\zeta/\alpha') - \eps$.
Setting $\alpha = \zeta^{1/2} \log^{3/4} n (\log\log n)^{1/2}$, we get that the accuracy of $\widehat{f}$ is at least $1-O(\zeta^{1/2} \log^{3/4} n (\log\log n)^{1/2}) - \eps$.
The running time is dominated by the at most $n$ executions of the algorithm from Lemma \ref{lem:pseudoregular_embedding_euclidean}, and thus it is at most $n^{O(1)} 2^{\eps^{-2} (\frac{d \log n}{\zeta \eps'})^{O(d)}}$, for any fixed $u>0$, concluding the proof.
\end{proof}

\section{Learning tree metric spaces with perfect information}
\label{sec:trees_perfect}

In this Section we present our algorithm for learning tree metric spaces, in the setting of perfect information.
We first show that if some instance to the problem admits an embedding into a tree of bounded diameter, then it also admits an embedding, with approximately optimal accuracy, and with low distortion, into a \emph{fixed} tree of bounded cardinality.
Using the algorithm from Section \ref{sec:pseudoregular}, we obtain an algorithm for computing such a near-optimal embedding.
The final algorithm is obtained by combining the above result with an efficient procedure for partitioning the input, such that each sub-instance admits a near-optimal solution into a tree of bounded diameter.

\subsection{Canonical tree embeddings and pseudoregular partitions}

We now define a family of ``canonical'' trees, and show that any instance that admits an optimal solution into a tree of bounded diameter, also admits a near-optimal solution into a canonical tree.

Let $\alpha > 0$, and $k,k'\in \mathbb{N}$.
We denote by $T_{\alpha, k, k'}$ the full $k'$-ary tree of combinatorial depth $k$ (that is, every non-leaf node has $k'$ children, and every root-to-leaf path contains $k$ edges), and with every edge having length $\alpha$.

\begin{lemma}[Existence of canonical tree embeddings]\label{lem:canonical_tree_existence}
Let $\eps, \eps'>0$.
Let $\gamma_C=(C, {\cal S}, {\cal D}, u, \ell)$ be an instance of the metric learning problem.
Let $T^*$ be a tree of diameter at most $\Delta$, for some $\Delta>0$.
Suppose that $\gamma_C$ admits an embedding $g^*:C\to V(T^*)$ with accuracy $r^*$.
Then, there exists a $(1+\eps')$-embedding $\widebar{g}:C\to V(T_{\eps'/2, \lceil\Delta/\eps' \rceil, \lceil 1/\eps\rceil})$, with accuracy at least $r^*-\eps$.
\end{lemma}

\begin{proof}
We begin by modifying $T^*$ so that the depth of each node is a multiple of $\eps'/2$, as follows.
Let $v^*\in V(T^*)$.
We consider $T^*$ as being rooted at $v^*$.
We consider each $v\in V(T^*)$ inductively, starting from $v^*$, and processing $v$ after its ancestor has been processed.
For each $v\in V(T^*)$, let $D_v = d_{T^*}(v^*, v)$.
If $D_v$ is not a multiple of $\eps'/2$, then let
$D'_v = (\eps'/2) \lfloor D_v / (\eps'/2) \rfloor$.
Let $u$ be the father of $v$ in the current tree.
Let $D_u'$ be the distance between $u$ and $v^*$ in the current tree.
By induction, we have $D'_u = (\eps'/2) \lfloor D_u / (\eps'/2) \rfloor$.
If $D'_u = D'_v$, then we contract the edge $\{u,v\}$, effectively identifying $v$ and $u$.
Otherwise, it must be $D'_u<D'_v$.
We set the length of $\{u,v\}$ to be $D'_v-D'_u$, and for every child $w$ of $v$, we set the length of $\{w,v\}$ to be $d_{T^*}(w,v) - (D_v-D'_v)$. 

Let $T'$ be the resulting tree after processing all nodes of $T^*$, and let $g':C\to V(T')$ be the induced embedding.
It is immediate by the above construction that every pairwise distance in $T^*$ changes by an additive factor of at most $\eps'$, and thus  $g'$ is a $(1+\eps')$-embedding with accuracy $r^*$.

Next, we modify $g'$ to obtain an embedding $g'':C\to T''$, where  $T''=T_{\eps'/2, \lceil\Delta/\eps' \rceil, \lceil 1/\eps\rceil}$.
We use the probabilistic method to show that the desired mapping $g''$ exists.
We construct a random mapping $h:V(T') \to V(T'')$ as follows.
We map the root of $T'$ to the root of $T''$.
For every other $v\in V(T')$, let $u$ be the parent of $v$ in $T'$.
Suppose that $h(u)$ has already been defined.
We set $h(v)$ to be some child $w$ of $u$ in $T''$, chosen uniformy at random from the set of $\lceil 1/\eps \rceil$ children of $u$.
This completes the construction of the random mapping $h:V(T') \to V(T'')$.
We set $g''= h \circ g'$, where $\circ$ denotes function composition.

It remains to show that $g''$ satisfies the assertion with positive probability.
Let $u,v\in C$.
If $g'(u)$ and $g'(v)$ lie in the same branch of $T'$, then we have, with probability 1, that 
\[
d_{T''}(g''(u), g''(v)) = d_{T'}(g'(u), g'(v)),
\]
and thus it suffices to consider $u,v\in C$ such that $g'(u)$ and $g'(v)$ lie in different branches of $T'$.
Let $a\in V(T')$ be the nearest common ancestor of $g'(u)$ and $g'(v)$ in $T'$.
Let $u', v'\in V(T')$ be the children of $a$ that are closest to $u$ and $v$ respectively (note that $u=u'$ if $u'$ is a child of $a$, and similarly for $v$).
We have
\begin{align*}
\Pr_h[d_{T''}(g''(u), g''(v)) = d_{T'}(g'(u), g'(v))] &\geq \Pr_h[h(u') \neq h(v')] = \frac{1}{\lceil 1/\eps\rceil} \leq \eps,
\end{align*}
where the equality follows by the fact that each of $h(u')$ and $h(v')$ is chosen uniformly at random from the set of children of $w$, which has cardinality $\lceil 1/\eps \rceil \geq 1/\eps$.
By the union bound on all pairs $\{u,v\}\in {\cal S}\cup {\cal D}$, we obtain that the expected number of constraints that are satisfied by $g'$ but violated by $g''$ is at most $\eps \cdot |{\cal S}\cup {\cal D}|$.
It follows that the expected accuracy of $g''$ is at least $r^* - \eps$.
This implies that there exists some desired $\widehat{g} = g'':C\to T''$, with accuracy at least $r^*-\eps$, which concludes the proof.
\end{proof}

Combining the above with the algorithm from Section \ref{sec:pseudoregular}, we obtain the following algorithm for learning tree metric spaces of bounded diameter.

\begin{lemma}[Computing canonical tree embeddings]\label{lem:canonical_tree_compute}
Let $\eps, \eps'>0$.
Let $\gamma_C=(C, {\cal S}, {\cal D}, u, \ell)$ be an instance of the metric learning problem.
Let $\alpha>0$, and $k,k'\in \mathbb{N}$.
Suppose that $\gamma_C$ admits a $(1+\eps')$-embedding $\widebar{g}:C\to V(T_{\alpha, k, k'})$ with accuracy $\widebar{r}$.
Then, there exists an algorithm which given $\gamma_C$ computes a $(1+\eps')$-embedding $\widehat{g}:C\to V(T_{\alpha,k,k'})$, with accuracy at least $\widebar{r}-\eps$.
Furthermore for any fixed $u$ and $\ell$, the running time is 
$n^{O(1)} 2^{O(k^{5k'}/\eps^2)}$.
\end{lemma}

\begin{proof}
The result follows from Lemma \ref{lem:pseudo_embed}, by setting $M=(N,\rho)$ to be the shortest-path metric of $T_{\alpha, k, k'}$.
The running time is $n^{O(1)} 2^{O(|N|^5/\eps^2)} = n^{O(1)} 2^{O(|V(T_{\alpha,k,k'})|^5/\eps^2)} = n^{O(1)} 2^{O(k^{5k'}/\eps^2)}$.
\end{proof}

\subsection{The algorithm}






We are now ready to prove the main result of this Section.

\begin{proof}[Proof of Theorem \ref{thm:trees_perfect_info}]
Let $G_{\cal S}$ be the graph with $V(G_{\cal S})=X$ and $E(G_{\cal S}) = {\cal S}$.
Let $\rho_S$ denote the shortest-path metric of $G_{\cal S}$, where we consider each edge having length $u$.

The algorithm starts by computing a random partition ${\cal P}$ of $X$, as follows.
Pick an arbitrary $v^*\in X$.
Pick $\alpha\in [0,1)$, uniformly at random.
For each $i\in \mathbb{N}$, let
\[
X_i = \left\{x\in X: \Delta\cdot (i+\alpha)/2\leq \rho_{\cal S}(v^*, x) < \Delta \cdot (i+\alpha+1)/2 \right\},
\]
where $\Delta = 8u/\eps$.
For each $i$, let ${\cal C}_i$ be the set of connected components of $G_{\cal S}[X_i]$.
Let ${\cal C} = \bigcup_{i\in \mathbb{N}} {\cal C}_i$.
In other words, ${\cal C}$ is the collection of connected components in the partition of $G_{\cal S}$ obtained by deleting all edges with endpoints in different sets $X_i$, $X_j$.

Let ${\cal S}' \subseteq {\cal S}$ be the set of edges of $G_{\cal S}$ that are cut by the partition ${\cal C}$ (that is, they have endpoints in different clusters in ${\cal C}$).
It is immediate that for each $\{x,y\}\in {\cal S}$,
\begin{align*}
\Pr[\{x,y\} \text{ is cut by } {\cal C}] &\leq |\rho_S(v^*, x) - \rho_S(v^*, y)| / (\Delta/2) \leq u/(\Delta/2) \leq \eps/4
\end{align*}
Therefore, by the linearity of expectation, the expected number of edges in ${\cal S}$ that are cut by ${\cal C}$ is at most $|{\cal S}| \eps /4$.

Consider $T^*$ as being rooted at $r^*=f^*(v^*)$.
We next show that that for any $x,y\in X$, which lie on the same branch of $T^*$, we have that 
\begin{align}
\rho_{\cal S}(v^*, x) < \rho_{\cal S}(v^*, y) \Rightarrow d_{T^*}(r^*, f^*(x)) < d_{T^*}(r^*, f^*(y)). \label{eq:T_induction0}
\end{align}
For any integer $i\geq 0$, let 
\[
V_i = \{x\in X : \rho_{\cal S}(v^*, x) = u \cdot i\}.
\]
Note that $V_0=\{v^*\}$.
Assume, for the sake of contradiction, that \eqref{eq:T_induction0} does not hold.
That is, there exists some branch $P$ of $T^*$, and some $x,y\in X$, with $f^*(x),f^*(y)\in P$, such that $\rho_{\cal S}(v^*, x) < \rho_{\cal S}(v^*, y)$, and $d_{T^*}(r^*, f^*(x)) > d_{T^*}(r^*, f^*(y))$.
It follows that $x\neq v^*$.
Let $i_1,i_2\geq 0$, such that $x\in V_{i_1}$ and $y\in V_{i_2}$.
Since $\rho_{\cal S}(v^*, x) < \rho_{\cal S}(v^*, y)$, we have $i_1<i_2$.
Since $x\neq v^*$, we have $i_1>0$.
Let $x'\in X$ be the neighbor of $x$ along a shortest $x$-$v^*$ path in $G_{\cal S}$.
By picking $x$ so that $d_{T^*}(f^*(x),r^*)$ is minimized, it follows that $f^*(y)$ lies on the path between $f^*(x)$ and $f^*(x')$ in $T^*$.
Since $\{x,x'\}\in {\cal S}$, and $f^*$ has accuracy $1$, it follows that 
\begin{align}
d_{T^*}(f^*(x), f^*(y)) &\leq u, \label{eq:T_induction1}
\end{align}
and
\begin{align}
d_{T^*}(f^*(x'), f^*(y)) &\leq u, \label{eq:T_induction2}
\end{align}
Using again the fact that $f^*$ has accuracy $1$, we get by \eqref{eq:T_induction1} and \eqref{eq:T_induction2} that $\{x,y\}\in {\cal S}$, and $\{x',y\}\in {\cal S}$.
This implies that $y\in V_{i_1}$, and thus $i_1=i_2$, a contradiction.
We have thus established \eqref{eq:T_induction0}.

It follows by \eqref{eq:T_induction0} that for any branch $P$ of $T^*$, when restricting ${\cal C}$ on $(f^*)^{-1}(V(P))$, we obtain subsets of $X$ that are mapped to disjoint subpaths of $P$.
Moreover, since $f^*$ has accuracy 1, we have that for all $\{x,y\}\in {\cal S}$, $d_{T^*}(f^*(x), f^*(y)) \leq u$, and thus each such subpath of $P$ has length at most $\Delta/2$.
Since the shortest path between any two vertices of any tree consists of at most two subpaths of distinct branches, it follows that for any $C\in {\cal C}$, 
\[
\diam_{T^*}(f^*(C)) \leq \Delta.
\]
In other words, by restricting $f^*$ on $C$, we obtain an embedding $f^*_C$ of $C$ with accuracy $1$, into some tree $T^*_C$ of diameter at most $\Delta$.
By Lemma \ref{lem:canonical_tree_existence} 
there exists some $(1+\eps')$-embedding $\widebar{f}_C$ of $C$ into  $T_{\eps'/2, \lceil \Delta/\eps'\rceil, \lceil 8/\eps\rceil}$, with accuracy at least $1-\eps/8$.
Therefore by Lemma \ref{lem:canonical_tree_compute} we can compute an embedding $\widehat{f}_C$ of $C$ into some tree $\widehat{T}_C$, with accuracy at least $1-\eps/8-\eps/8=1-\eps/4$, in time $n^{O(1)} 2^{O((\lceil \Delta/\eps'\rceil)^{5(\lceil 2/\eps\rceil)}/\eps^2)}$, with high probability.
For any fixed $u>0$, this running time is $n^{O(1)} 2^{1/(\eps\eps')^{O(1/\eps)}}$.

Finally, we can merge all the trees $\widehat{T}_C$, for all $C\in {\cal C}$, into a single tree $\widehat{T}$, as follows.
We introduce a new vertex $\widehat{r}$, and we connect $\widehat{r}$ to an arbitrary vertex $u_C$ in each $C$, via an edge $\{\widehat{r}, u_C\}$ of length $2\ell$.
We set $\widehat{f}$ to be the resulting embedding of $X$ into $\widehat{T}$.
The total running time for computing $\widehat{T}$ is dominated by the computation of the embeddings $\widehat{f}_C$, for all $C\in {\cal C}$.
Since there are at most $n$ clusters in ${\cal C}$, it follows that the total running time is $n^{O(1)} 2^{1/(\eps\eps')^{O(1/\eps)}}$, as required.

By the choice of the lenght of the new edges added to $\widehat{T}$, all edges in ${\cal D}$ that are cut by ${\cal C}$ are satisfied in the resulting embedding $\widehat{f}$.
It follows that $\widehat{f}$ has expected accuracy at least $1-\eps/2$.
By repeating the algorithm $O(\log n)$ times and returning the best solution found, we obtain an embedding with accuracy at least $1-\eps$, with high probability, which concludes the proof.
\end{proof}

\section{Learning tree metric spaces with imperfect information}
\label{sec:trees_imperfect}

This Section presents our algorithm for learning tree metric spaces under imperfect information.
The scructure of the algorithm and its proof of correctness resemble the arguments from the Euclidean case.

\begin{proof}[Proof of Theorem \ref{thm:trees_imperfect_info}]
Fix some optimal embedding $f^*:X\to V(T^*)$, for some tree $T^*$, with accuracy $1-\zeta$.
Let $E'\subset {\cal S} \cup {\cal D}$ be the set of constraints computed by the algorithm in Lemma \ref{lem:well-linked}.
We have 
$|E'| \leq \alpha |{\cal S}\cup {\cal D}|$,
for some $\alpha>0$ to be determined.

For each connected component $C$ of $G_{{\cal S}} \setminus E'$, let $\gamma_C=(C,{\cal S}_C,{\cal D}_C,u,\ell)$ be the restriction of $\gamma$ on $C$.
Let $f^*_C:C\to V(T^*)$ be the restriction of $f^*$ on $C$, and let the accuracy of $f^*_C$ be $1-\zeta_C$, for some $\zeta_C\in [0,1]$.
Let $F_C$ be the set of constraints in $E_{\cal S}(C)$ that are violated by $f^*_C$.
By Lemma \ref{lem:isoperimetry} it follows that there exists an embedding $f'_C:X\to V(T^*)$, with accuracy $1-O(\zeta_C/\alpha')$, and such that $\diam(f'_C(C)) \leq \Delta := O(\frac{u\log n}{\alpha'})$, for some $\alpha'=\Omega(\frac{\alpha}{\log^{3/2} n \log\log n})$.

Arguing as in the proof of Theorem \ref{thm:trees_perfect_info},
by Lemma \ref{lem:canonical_tree_existence} 
there exists some $(1+\eps')$-embedding $\widebar{f}_C$ of $C$ into  $T_{\eps'/2, \lceil \Delta/\eps'\rceil, \lceil 8/\eps\rceil}$, with accuracy at least $1-O(\zeta_C/\alpha')-\eps/8$.
Therefore by Lemma \ref{lem:canonical_tree_compute} we can compute an embedding $\widehat{f}_C$ of $C$ into some tree $\widehat{T}_C$, with accuracy at least $1-O(\zeta_C/\alpha')-\eps/8-\eps/8=1-O(\zeta_C/\alpha')-\eps/4$, in time $n^{O(1)} 2^{O((\lceil \Delta/\eps'\rceil)^{5(\lceil 2/\eps\rceil)}/\eps^2)}$, with high probability.
We can merge all the trees $\widehat{T}_C$, for all $C\in {\cal C}$, into a single tree $\widehat{T}$, as follows.
We introduce a new vertex $\widehat{r}$, and we connect $\widehat{r}$ to an arbitrary vertex $u_C$ in each $C$, via an edge $\{\widehat{r}, u_C\}$ of length $2\ell$.
Let $\widehat{f}:\to V(\widehat{T})$ be the resulting embedding.

By the choice of the lenght of the new edges added to $\widehat{T}$, all edges in ${\cal D}$ that are cut by ${\cal C}$ are satisfied in the resulting embedding $\widehat{f}$.
It follows that $\widehat{f}$ has expected accuracy at least $1-\alpha-O(\zeta/\alpha')-\eps/2$.
By repeating the algorithm $O(\log n)$ times and returning the best solution found, we obtain an embedding with accuracy at least $1-\alpha-O(\zeta/\alpha')-\eps$, with high probability.
Setting $\alpha = \zeta^{1/2} \log^{3/4} n (\log\log n)^{1/2}$, we get that the accuracy of $\widehat{f}$ is at least $1-O(\zeta^{1/2} \log^{3/4} n (\log\log n)^{1/2}) - \eps$, as required.
The total running time is dominated by the computation of the embeddings $\widehat{f}_C$, for all $C\in {\cal C}$.
Since there are at most $n$ clusters in ${\cal C}$, it follows that the total running time is $n^{O(1)} 2^{O((\lceil \Delta/\eps'\rceil)^{5(\lceil 2/\eps\rceil)}/\eps^2)}$, which for any fixed $\zeta$ and $u$, is at most $2^{\left(\frac{u \log n}{\zeta^{1/2}\eps \eps'}\right)^{O(1/\eps)}}$, as required.
\end{proof}

\section{Learning line metric spaces with perfect information}
\label{sec:line_perfect}

This Section is devoted to obtaining an algorithm for learning a line metric space with perfect information, culminating to the proof of Theorem \ref{thm:perfect_info_line}.
For the remainder of this Section, 
let $\gamma=(X, {\cal S}, {\cal D}, u, \ell)$ be an instance of the problem of learning a line metric space, with $|X|=n$.
Let $G_{\cal S}$ be the graph with $V(G_{\cal S})=X$ and $E(G_{\cal S}) = {\cal S}$.
We may assume w.l.o.g.~that $G_{\cal S}$ is connected, since otherwise we can solve the problem on each connected component independently, and concatenate the resulting embeddings by leaving a gap of length $\ell$ between the images of different components.

For a mapping $f:X\to \mathbb{R}$ and some ordering $\sigma=x_1,\ldots,x_n$ of $X$, we say that $f$ is \emph{compatible} with $\sigma$ if for all $x_i,x_j\in X$, if $f(x_i)<f(x_j)$, then $i<j$.
For some $m\in [n]$, and define the prefix $\sigma[m]=x_1,\ldots,x_m$.
We say that $f$ is compatible with $\sigma[m]$ if the above condition holds for all $i,j\in [m]$, and for all $k\in \{m+1,\ldots,n\}$, we have $f(x_m)\leq f(x_k)$.

\begin{lemma}[Computing a set of permutations]\label{lem:line_permutation}
Suppose that $G_{\cal S}$ is connected.
Then there exists a polynomial-time algorithm which given $\gamma$, computes a set of total orderings $\sigma_1,\ldots,\sigma_n$ of $X$, satisfying the following condition:
If $\gamma$ admits an embedding with accuracy $1$, then there exists some $i\in [n]$, such that $\gamma$ admits an embedding with accuracy $1$ that is compatible with $\sigma_i$.
\end{lemma}

\begin{proof}
We may assume w.l.o.g.~that $|X|>3$, since otherwise we can simply output all possible total orderings of $X$.
Fix some embedding $f$ with accuracy $1$.
We compute a set $\sigma_1,\ldots,\sigma_n$ of total orderings of $X$, where each $\sigma_i$ starts with a distinct element in $X$.
It therefore suffices to describe the construction of $\sigma_i=x_{i,1},\ldots,x_{i,n}$, for some $i\in [n]$, assuming that $x_{i,1}$ is fixed.
The algorithm computes a total ordering $x_{i,1},\ldots,x_{i,n}$ inductively.
Since we try all possible starting elements, it follows that for some $i$, $f$ is compatible with $\sigma_i[1]$.
For this value of $i$, we will maintain the inductive invariant that $f$ is compatible with $\sigma_i[j]$, for all $j$.
Suppose that $x_{i,1},\ldots,x_{i,j}$ have already been computed for some $j\in [n-1]$.
By the inductive invariant, we have that $f$ is compatible with $\sigma_i[j]$ (for some $i$).
Let $G'$ be the graph obtained from $G_{\cal S}$ by contracting $\{x_{i,1},\ldots,x_{i,j}\}$ into a single vertex $x'$.
Since $f$ is compatible with $\sigma_i[j]$, it follows that there exists an embedding with accuracy $1$ for $G'$ (that is, with the set of similar pairs being equal to the set of edges of $G'$), with $x'$ being the left-most vertex; such an embedding can be obtained by restricting $f$ on $X\setminus \{x_{i,1},\ldots,x_{i,j-1}\}$.
Thus the neighborhood of $x'$ in $G'$ is a clique $C'$.
Let $C''\subset C'$ be the set of all vertices in $C'$ that have minimum degree in $G'$.
It follows that $C''$ forms an orbit in the group of automorphisms of $G'$ (since any two vertices in $C''$ must have the same neighborhood by the fact that $f$ has accuracy $1$).
Therefore we may assume w.l.o.g.~that the vertex that is embedded immediately to the right of $x'$ is any of the vertices in $C''$.
We set $x_{i,j+1}$ to be an arbitrary vertex in $C''$.
This concludes the construction.
By the inductive hypothesis, there exists some $i\in \{1,\ldots,n\}$, such that $f$ is compatible with $\sigma_i[n]=\sigma_i$, which concludes the proof.
\end{proof}

\begin{lemma}[From a permutation to an embedding]\label{lem:line_embedding}
There exists a polynomial-time algorithm which given $\gamma$, and a total ordering $\sigma=x_1,\ldots,x_n$ of $X$, terminates with one of the following outcomes:
\begin{description}
\item{(1)}
Computes some embedding $f:X\to \mathbb{R}$ with accuracy $1$.
\item{(2)}
Correctly decides that there exists no embedding $f:X\to \mathbb{R}$, with accuracy $1$, that is compatible with $\sigma$.
\end{description}
\end{lemma}

\begin{proof}
We argue that the problem of deciding whether there exists an embedding $f$ with accuracy $1$ that is compatible with $\sigma$, can be reduced to the feasibility of the following linear program:
\begin{align*}
\begin{array}{ll}
  \alpha_i \in \mathbb{R} & \text{ for all } i \in [n-1]\\
  \sum_{t=i}^{j-1} \alpha_t \leq u & \text{ for all } \{x_i,x_j\}\in {\cal S}, i<j\\
  \sum_{t=i}^{j-1} \alpha_t \geq \ell & \text{ for all } \{x_i,x_j\}\in {\cal D}, i<j
\end{array}
\end{align*}

If there exists an embedding $f:X\to \mathbb{R}$ with accuracy $1$, that is compatible with $\sigma$, then we can find a feasible solution for the LP by setting $\alpha_i=f(x_{i+1})-f(x_i)$.
It follows that for all $\{x_i,x_j\} \in {\cal D}\cup {\cal S}$, with $i<j$, we have (by the compatibility of $f$ and $\sigma$) that 
\[
|f(x_i)-f(x_j)|=\sum_{t=i}^{j-1} \alpha_t, 
\]
and thus this assignment forms a feasible solution for the LP.

It remains to show that if the LP is feasible, then there exists an embedding $f:X\to \mathbb{R}$ with accuracy $1$.
Such an embedding can be found by setting $f(x_1)=0$, and for all $i>1$, $f(x_i)=\sum_{t=1}^{i-1} \alpha_t$, for some feasible solution $\{a_i\}_i$.
The feasibility of the LP immediately implies that $f$ has accuracy $1$.
\end{proof}

\begin{proof}[Proof of Theorem \ref{thm:perfect_info_line}]
Let $G_{\cal S}$ be the graph defined above.
Let ${\cal C}$ be the set of connected components of $G_{\cal S}$.
For each $C\in {\cal C}$, let $\gamma_C$ be the subproblem obtained by restricting $\gamma$ on induced on $C$; that is 
\[
\gamma_C = \left(X\cap C, {\cal S}\cap \binom{C}{2}, {\cal D}\cap \binom{C}{2}, u, \ell\right).
\]
We can obtain a solution for $\gamma$ by solving each subproblem $\gamma_C$ independently and combining the solutions by placing the images of embedding for different components at distance at least $\ell$ from each other.

It therefore remains to obtain an algorithm for the case where $G_{\cal S}$ is connected.
This can be done by first computing, using Lemma \ref{lem:line_permutation}, a set of total orderings, $\sigma_1,\ldots,\sigma_n$.
For each such ordering $\sigma_i$, we run the algorithm of \ref{lem:line_embedding}.
If the algorithm returns some embedding $f_i$, then we output $f_i$; if no execution of the algorithm outputs an embedding, then we can correctly decide that no embedding with accuracy $1$ exists, which concludes the proof.
\end{proof}

\section{NP-hardness of learning the plane}
\label{sec:np-hard}

In this Section we show that the problem of learning a 2-dimensional Euclidean metric space is NP-hard, even in the case of perfect information.

\begin{proof}[Proof of Theorem \ref{thm:np-hard}]
A graph $G$ is called \emph{unit disk} if it can be realized as the intersection graph of a set of unit disks in the plane.
The problem of deciding whether a given graph is unit disk is NP-hard \cite{breu1998unit}.
Furthermore, the reduction in \cite{breu1998unit} shows that it is NP-hard, given some graph $G$, to decide whether there exists some $f:V(G)\to \mathbb{R}^2$, such that for all $\{u,v\}\in E(G)$, $\|f(u)-f(v)\|_2\leq 1$, and for all $\{u,v\}\notin E(G)$, $\|f(u)-f(v)\|_2\geq 1+\alpha$, for some fixed $\alpha>0$.
It is easy to see that this problem is a special case of the non-linear metric learning problem, by setting $G=(X,{\cal S})$, $d=2$, $u=1$, and $\ell=1+\alpha$, which concludes the proof.
\end{proof}

\bibliographystyle{alpha}
\bibliography{bibfile}

\end{document}